\def\final{0}
\def\podscamera{0}

\ifnum\podscamera=1
	\documentclass{sig-alternate-2013}
\else
	\documentclass[10pt]{article}
\fi

\ifnum\podscamera=0
\usepackage{algorithm, algorithmic}
\usepackage{amsmath, amssymb, amsthm}
\usepackage{fullpage}
\usepackage[margin=1.1in]{geometry}
\usepackage{verbatim}
\usepackage{color}
\usepackage{tabulary}
\definecolor{DarkGreen}{rgb}{0.1,0.5,0.1}
\definecolor{DarkRed}{rgb}{0.5,0.1,0.1}
\definecolor{DarkBlue}{rgb}{0.2,0.2,0.6}
\usepackage[pdftex]{hyperref}
\hypersetup{
    unicode=false,				% non-Latin characters in Acrobat¿s bookmarks
    pdftoolbar=true,				% show Acrobat toolbar?
    pdfmenubar=true,			% show Acrobat menu?
    pdffitwindow=false, 		% page fit to window when opened
    pdftitle={},    					% title
    pdfauthor={}
    pdfsubject={},				% subject of the document
    pdfnewwindow=true,		% links in new window
    pdfkeywords={},			% list of keywords
    colorlinks=true,				% false: boxed links; true: colored links
    linkcolor=DarkBlue,		% color of internal links
    citecolor=DarkBlue,	% color of links to bibliography
    filecolor=DarkRed,		% color of file links
    urlcolor=DarkBlue,		% color of external links
}
\usepackage{framed}
\usepackage{kpfonts}
\usepackage{url}
\fi
\ifnum\podscamera=1
\usepackage{verbatim}
\usepackage{algorithm, algorithmic}
\usepackage{amsmath, amssymb}
\usepackage{tabulary}
\usepackage{framed}
\fi

\ifnum\final=0
\newcommand{\mynote}[1]{\marginpar{\tiny\sf #1}}
\else
\newcommand{\mynote}[1]{}
\fi

\newcommand\N{\mathbb{N}}
\newcommand\R{\mathbb{R}}

\newcommand\cA{\mathcal{A}}
\newcommand\cB{\mathcal{B}}

\newcommand\cQ{\mathcal{Q}}
\newcommand\cR{\mathcal{R}}
\newcommand\cX{\mathcal{X}}

\newcommand\cL{\mathcal{L}}

\newcommand\poly{\mathrm{poly}}
\newcommand{\polylog}{\mathrm{polylog}}
\newcommand\bits{\{0,1\}}

\newcommand{\getsr}{\gets_{\mbox{\tiny R}}}
\newcommand\set[1]{\left\{#1\right\}} %usage \set{1,2,3,,}
\newcommand{\from}{:}

\newcommand{\grad}{\nabla}

\DeclareMathOperator*{\Expectation}{\mathbb{E}}
\newcommand{\Ex}[2]{\Expectation_{#1}\left(#2\right)}
\DeclareMathOperator*{\Probability}{\mathrm{Pr}}
\newcommand{\prob}[1]{\mathrm{Pr}\left(#1\right)}
\newcommand{\Prob}[2]{\Probability_{#1}\left(#2\right)}
\DeclareMathOperator*{\argmin}{\mathrm{argmin}}

\newcommand{\INDSTATE}[1][1]{\STATE\hspace{#1\algorithmicindent}}

\newcommand{\univ}{\cX}
\newcommand{\eps}{\varepsilon}
\newcommand{\db}{D}
\newcommand{\rows}{n}
\newcommand{\row}{x}
\newcommand{\mech}{\cA}
\newcommand{\loss}{\ell}
\newcommand{\dom}{\Theta}
\newcommand{\param}{\theta}
\newcommand{\query}{q}
\newcommand{\paramdimension}{d}
\newcommand{\privparamt}{\param}
\newcommand{\acc}{\alpha}
\newcommand{\queryset}{\cQ}
\newcommand{\lossset}{\cL}

\newcommand{\err}{\mathrm{err}}
\newcommand{\losst}{\loss}
\newcommand{\paramt}{\hat{\param}}
\newcommand{\scale}{S}
\newcommand{\dbt}{\hat{\db}}

\ifnum\podscamera=1
\newtheorem{theorem}{Theorem}
\newtheorem{lemma}{Lemma}
\newtheorem{claim}{Claim}

\newdef{definition}{Definition}
\else
\newtheorem{theorem}{Theorem}[section]
\newtheorem{lemma}[theorem]{Lemma}

\newtheorem{claim}[theorem]{Claim}

\theoremstyle{definition}
\newtheorem{definition}[theorem]{Definition}

\fi

\ifnum\podscamera=1
	\newfont{\mycrnotice}{ptmr8t at 7pt}
	\newfont{\myconfname}{ptmri8t at 7pt}

	\permission{Permission to make digital or hard copies of all or part of this work for personal or classroom use is granted without fee provided that copies are not made or distributed for profit or commercial advantage and that copies bear this notice and the full citation on the first page. Copyrights for components of this work owned by others than ACM must be honored. Abstracting with credit is permitted. To copy otherwise, or republish, to post on servers or to redistribute to lists, requires prior specific permission and/or a fee. Request permissions from permissions@acm.org.}
	\conferenceinfo{PODS'15,}{May 31--June 4, 2015, Melbourne, Victoria, Australia.}
	\copyrightetc{Copyright \copyright~2015 ACM \the\acmcopyr}
	\crdata{978-1-4503-2757-2/15/05\ ...\$15.00.\\ http://dx.doi.org/10.1145/2745754.2745755}

\begin{document}
	\title{Private Multiplicative Weights Beyond Linear Queries%\titlenote{A full version of this paper appears online at arXiv.org~\cite{Ullman14}}
	}

	\numberofauthors{1} 
	\author{
		\alignauthor Jonathan Ullman\titlenote{The majority of this work was done while the author was at the Center for Research on Computation and Society at Harvard University.  Supported by NSF grant CNS-1237235.}\\
      		 \affaddr{Columbia University}\\
       		\affaddr{New York, NY}\\
      		 \email{jullman@cs.columbia.edu}
	}
	
	\maketitle
\else
	\title{Private Multiplicative Weights Beyond Linear Queries}

	\author{Jonathan Ullman\ifnum\final=1\thanks{The majority of this work was done while the author was a Postdoctoral Fellow at the Center for Research on Computation and Society at Harvard University.  Supported by NSF grant CNS-1237235.}\fi \\
	\ \\
	Department of Computer Science\\
	Columbia University, New York, NY. \\
	\href{mailto:jullman@cs.columbia.edu}{jullman@cs.columbia.edu}}
	\begin{document}
	\maketitle
\fi

\begin{abstract}
A wide variety of fundamental data analyses in machine learning, such as linear and logistic regression, require minimizing a convex function defined by the data.  Since the data may contain sensitive information about individuals, and these analyses can leak that sensitive information, it is important to be able to solve convex minimization in a privacy-preserving way.  

A series of recent results show how to accurately solve a single convex minimization problem in a differentially private manner.  However, the same data is often analyzed repeatedly, and little is known about solving multiple convex minimization problems with differential privacy.  For simpler data analyses, such as linear queries, there are remarkable differentially private algorithms such as the private multiplicative weights mechanism (Hardt and Rothblum, FOCS 2010) that accurately answer exponentially many distinct queries.  In this work, we extend these results to the case of convex minimization and show how to give accurate and differentially private solutions to \emph{exponentially many} convex minimization problems on a sensitive dataset.
\end{abstract}

\ifnum\podscamera=0
\vfill
\newpage

\tableofcontents
\vfill
\newpage
\fi

\section{Introduction}
Consider a dataset $D = (x_1,\dots,x_n) \in \univ^n$ in which each of the $n$ rows corresponds to an individual's record, and each record consists of an element of some data universe $\univ$.  The goal of privacy-preserving data analysis is to enable rich statistical analyses on such a dataset while protecting the privacy of the individuals.  It is especially desirable to achieve \emph{differential privacy}~\cite{DworkMNS06}, which guarantees that no individual's data has a significant influence on the information released about the dataset.  

In this work we consider differentially private algorithms that answer \emph{convex minimization (CM) queries} on the sensitive dataset.  A CM query is specified by a convex \emph{loss function} $\ell \from \dom \times \univ \to \R$, where $\dom$ is a convex set, and the corresponding \emph{query} $q_{\ell} \from \univ^* \to \dom$ selects the point $\param \in \dom$ that minimizes the average loss on the rows of $D$.  That is,
$$
\query_{\loss}(D) = \argmin_{\param \in \dom} \frac{1}{n} \sum_{i=1}^{n} \ell(\param; x_i).
$$
These queries capture fundamental data analyses such as linear and logistic regression and support vector machines.  For example, we may have a dataset consisting of $n$ labeled examples $(x_1,y_1), \dots, (x_n, y_n)$ from the data universe $\univ = \R^d \times \R$ (corresponding to $d$ attributes and a single label per individual), and wish to compute the linear regression
$$
\param^* = \argmin_{\param \in \R^d} \frac{1}{n} \sum_{i=1}^{n} \left(\langle \param, x_i \rangle - y_i \right)^2
$$

Starting with the results of Dwork and Lei~\cite{DworkL09} and Chaudhuri, Monteleone, and Sarwate~\cite{ChaudhuriMS11}, there has been a long line of work~\cite{KiferST12, SmithT13,JainT14,BassilyST14} showing how to compute an accurate and differentially private answer to a single CM query.  However, in practice the same sensitive dataset will be analyzed by many different analysts, and together these analysts will need answers to a large number of distinct CM queries on the dataset.  Any algorithm for solving a single CM query can be applied repeatedly to answer multiple CM queries using the well known composition properties of differential privacy.  However, this straightforward approach incurs a significant loss of accuracy, and renders the answers meaningless after a small number of queries (roughly $n^2$ in most natural settings).

Fortunately, for many interesting types of queries, there are remarkable differentially private algorithms~\cite{BlumLR08, DworkNRRV09, DworkRV10, RothR10, HardtR10, GuptaRU12, HardtLM12} that are capable of giving accurate answers to \emph{exponentially many} different queries---far greater than what can be achieved using straightforward composition.  The most extensively studied case is \emph{linear queries}, which are specified by a property $p$ and ask ``What fraction of rows in $\db$ satisfy $p$?''  It is also known how to answer exponentially many arbitrary Lipschitz, real-valued queries~\cite{DworkRV10}, which generalize linear queries.  There are, however, no known non trivial algorithms for privately and accurately answering large sets of CM queries.%\footnote{Although CM queries are neither real-valued nor low-sensitivity, the techniques in~\cite{DworkRV10} can be adapted to answer exponentially many CM queries.  However, doing so would incur running time exponential in the size of the dataset $n$ and achieve quantitatively weaker utility guarantees than what we obtain in this work.}

In this work we show for the first time that it is possible to give accurate and differentially private answers to exponentially many convex minimization queries.  We do so via an extension of the simple and elegant private multiplicative-weights framework of Hardt and Rothblum~\cite{HardtR10}, which is known to achieve asymptotically optimal worst-case accuracy~\cite{BunUV14} and worst-case running time~\cite{Ullman13} for answering large families of linear queries.  Moreover, private multiplicative weights was shown to have a number of practical advantages~\cite{HardtLM12}, including good accuracy and running time in practice on low-dimensional datasets, parallelism, and simple implementation, all of which are preserved by our extension.  We believe that our technique for adapting the private multiplicative weights framework beyond linear queries may be useful in the future design of differentially private algorithms for other types of non linear queries.

\subsection{Our Results}

We can now state our results for answering large numbers of CM queries.  In order to answer even a single CM query, we need to place some sort of restrictions on the loss function $\loss$.  In particular, we consider the following types of restrictions on $\ell$:
\begin{itemize}
\item Lipschitz. $\| \grad \ell(\param; \row) \|_2 \leq 1$ for every $\param \in \dom, \row \in \univ$ (where the gradient is taken with respect to $\param$ for fixed $\row$).
\item $d$-Bounded.  $\dom \subseteq \{ \param \in \R^{d} \mid \|\param\|_2 \leq 1\}$.
\item $\sigma$-Strongly Convex.  $\loss(\param';\row) \geq \loss(\param; \row) + \langle \grad \loss(\param; \row), \param' - \param \rangle + \frac{\sigma}{2} \|\param' - \param\|_2^2$ for every $\param, \param' \in \dom, \row \in \univ$ (where, again, the gradient is taken with respect to $\param$ for fixed $\row$).
\item Unconstrained Generalized Linear Models (UGLM).  $\dom = \R^d$, $\univ \subseteq \R^d$ and $\loss(\param; \row) = \loss'(\langle \param, \row \rangle)$ for a convex function $\loss' \from \R \to \R$.
\end{itemize}
The constant $1$ in the the Lipschitz and boundedness conditions is arbitrary.  One can obtain more general statements in terms of these parameters by rescaling.  For simplicity, we will assume throughout that all loss functions $\loss$ are differentiable, and thus will freely use the gradient operator.  However, for all our algorithms and theorems, the assumption that $\ell$ is differentiable is unnecessary and $\grad \loss$ can be replaced with an arbitrary subgradient of $\loss$.

Table 1 summarizes our results for these different restrictions on the loss functions.  In all cases our algorithms are interactive.  They take a dataset $\db \in \univ^{\rows}$ as input, interact with a data analyst who chooses a sequence of loss functions $\loss^1,\dots, \loss^k$, and return answers $\hat{\param}^{1}, \dots, \hat{\param}^{k} \in \dom$ such that for every $j = 1,\dots,k$
$$
\frac{1}{n} \sum_{i=1}^{n} \loss^j(\hat{\param}^{j}; x_{i}) \leq \left(\min_{\param \in \dom} \frac{1}{n} \sum_{i=1}^{n} \loss^j(\param; x_{i})\right) + \alpha
$$
for some error parameter $\alpha$.  %In fact our algorithm outputs a \emph{synthetic dataset} $\hat{\db} \in \univ^n$ such that $\hat{\param}_{1}, \dots, \hat{\param}_{|\lossset|}$ are the answers to each CM query in $\lossset$ on $\hat{\db}$.  
We note that the data analyst may be \emph{adaptive}, meaning the choice of $\loss_j$ can depend on the previous losses and answers $\loss^1,\hat{\param}^1,\dots,\loss^{j-1}, \hat{\param}^{j-1}$.  %We use $\lossset = \set{\loss_1,\dots,\loss_k}$ to denote the set of queries asked by the analyst, with $|\lossset| = k$.
Differential privacy becomes easier to achieve as $n$ becomes larger.  Thus, we ask how big $n$ has to be to achieve a given level of accuracy $\alpha$ for answering $k$ queries from a family of loss functions $\lossset$.  

Our results are summarized in the following table.  We emphasize that if one were to use an algorithm for answering a single CM query repeatedly via composition, then required database size $n$ would depend polynomially on $k$, whereas the error depends only polylogarithmically on $k$ in each of our results.

\begin{table}[t!] 
\begin{center}
\ifnum\podscamera=0
\begin{tabular}{|m{4cm}|m{4.5cm}|m{6cm}|}
\hline
Restrictions & $n$ Needed for a Single Query  & $n$ Needed for $k$ Queries \\
\hline
Linear Queries & $O\left( \frac{1}{\alpha} \right)$ \cite{DworkMNS06} & $\tilde{O}\left( \frac{\sqrt{\log |\univ|} \cdot \log k}{\alpha^2}\right)$ \cite{HardtR10}\\
\hline
Lipschitz, $d$-Bounded & $\tilde{O}\left( \frac{\sqrt{d}}{\alpha} \right)$ \cite{BassilyST14} & $\tilde{O}\left(\max\left\{ \frac{\sqrt{d \cdot \log |\univ|}}{\alpha^2}, \frac{ \log k \cdot \sqrt{\log |\univ|}}{\alpha^2}\right\} \right)$ \\
\hline
Lipschitz, $d$-Bounded, UGLM & $\tilde{O}\left( \frac{1}{\alpha^2} \right)$ \cite{JainT14} & $\tilde{O}\left( \max\left\{ \frac{\sqrt{\log |\univ|}}{\alpha^3}, \frac{\log k \cdot \sqrt{\log |\univ|}}{\alpha^2}\right\} \right)$ \\
\hline
Lipschitz, $d$-Bounded, $\sigma$-Strongly Convex & $\tilde{O}\left( \sqrt{\frac{d}{\sigma \alpha}} \right)$ \cite{BassilyST14} & $\tilde{O}\left( \max\left\{ \sqrt{\frac{d \cdot \log |\univ|}{\sigma \alpha^{3}}}, \frac{\log k \cdot \sqrt{\log |\univ|} }{\alpha^2}\right\} \right)$ \\
\hline
\end{tabular}
\else
\begin{tabular}{|m{1.7cm}|m{1.8cm}|m{3.6cm}|}
\hline
Restrictions & $n$ needed for one query   &  $n$ needed for $k$ queries \\
\hline
Linear Queries & $O\left( \frac{1}{\alpha} \right)$ \cite{DworkMNS06} & $\tilde{O}\left( \frac{\sqrt{\log |\univ|} \cdot \log k}{\alpha^2}\right)$ \cite{HardtR10}\\
\hline
Lipschitz, $d$-Bounded & $\tilde{O}\left( \frac{\sqrt{d}}{\alpha} \right)$ \cite{BassilyST14} & $\tilde{O}\Big(\max\Big\{\frac{\sqrt{d \cdot \log |\univ|}}{\alpha^2},$ $~~~~~~~~\frac{ \log k \cdot \sqrt{\log |\univ|}}{\alpha^2}\Big\} \Big)$ \\
\hline
Lipschitz, $d$-Bounded, UGLM & $\tilde{O}\left( \frac{1}{\alpha^2} \right)$ \cite{JainT14} & $\tilde{O}\Big( \max\Big\{  \frac{\sqrt{\log |\univ|}}{\alpha^3},$ $~~~~~~~~\frac{\log k \cdot \sqrt{\log |\univ|}}{\alpha^2}\Big\} \Big)$ \\
\hline
Lipschitz, $d$-Bounded, $\sigma$-Strongly Convex & $\tilde{O}\Big( \sqrt{\frac{d}{\sigma \alpha}} \Big)$ \cite{BassilyST14} & $\tilde{O}\Big( \max\Big\{  \sqrt{\frac{d \cdot \log |\univ|}{\sigma \alpha^{3}}},$ $~~~~~~~~\frac{\log k \cdot \sqrt{\log |\univ|} }{\alpha^2}\Big\} \Big)$ \\
\hline
\end{tabular}
\fi
\caption{\footnotesize Accuracy guarantees for answering various families of CM queries under differential privacy.  New results are shown in green.  Error bounds for linear queries, which are a special case of Lipschitz, 1-bounded CM queries are shown for comparison.  Error bounds for answering a single CM query under each restriction is also shown for comparison.  All results are stated for $(\eps, \delta)$-differential privacy for $\eps$ constant and $\delta$ a negligible function of $n$.}
\end{center}
\label{table:results}
\end{table} 

%To highlight the novel techniques, in the technical sections of the paper, we will give an algorithm for the simpler case in which all of $\lossset$ is specified in advance by the analyst.  However, such an algorithm can be transformed into an algorithm that allows for an online stream of adaptively chosen queries using standard techniques~\cite{RothR10, HardtR10, GuptaRU12}.

%We remark that our algorithm can also be made into a differentially private \emph{online algorithm} that receives an adaptively and adversarially chosen sequence of CM queries $\query_1,\dots,\query_{k}$ and answers each one as it arrives, with the same accuracy guarantees up to constant factors.  This transformation can be achieved using standard techniques~\cite{RothR10, HardtR10, GuptaRU12}.  We omit the details of this transformation, as these techniques are orthogonal to the contributions of this work.

Our algorithms have running time $\poly(n, |\univ|, k)$ assuming oracle access to $\loss$ and its gradient for every $\loss$.  Thus, our algorithms are not generally efficient, as $|\univ|$ will often be exponential in the dimensionality of the data.  For example, if $\univ = \bits^{d}$, then the dataset consists of $nd$ bits yet our algorithms run in time $2^d$, even when $k$ is polynomial and every loss function and its gradient can be efficiently computed.  Unfortunately this exponential running time is inherent, under widely believed cryptographic assumptions.  Even answering $n^{2+o(1)}$ linear queries, which are a special case of Lipschitz, $1$-Bounded CM queries, requires exponential time~\cite{Ullman13}.  

Additionally, our algorithms require significantly more error than answering a single CM query.  For example, in the case of Lipschitz, $d$-Bounded CM queries, a single query can be answered with a dataset of size $n = \tilde{O}(\sqrt{d}/\alpha)$, whereas answering $\poly(n)$ queries with our algorithm requires a dataset of size $n = \tilde{O}(\sqrt{\log |\univ|} \cdot \log k / \alpha^2)$.  By the results of Kasiviswanathan, Rudelson, and Smith~\cite{KasiviswanathanRS13}, a database of size at least $n = \Omega(1/\alpha^2)$ is necessary when answering $\gg 1/\alpha^2$ queries.  See Section~\ref{sec:complexity} for a more detailed discussion of the lower bounds and computational complexity issues that arise.

Since the error bounds and running time of our algorithm both depend on $|\univ|$, our error guarantees may appear vacuous when $\univ$ is infinite.  For example, in many common applications $\univ = \set{\theta \in \R^d \mid \|\theta\|_2 \leq 1}$ is the $d$-dimensional unit ball.  However, in many settings it is essentially without loss of generality (up to, say, a factor of $2$ in the error) to round the data points to some finite, data universe.  Typically if the data points lie in a $d$-dimensional space, the size of such a data universe will be $(d/\alpha)^{O(d)}$.  We leave it for future work to find algorithms that apply to continuous data universes in a more natural way.

\subsection{Techniques}

In order to describe our algorithms, it will be helpful to start by sketching the private multiplicative weights framework of Hardt and Rothblum~\cite{HardtR10} for answering linear queries.  Here, we focus on the ``offline'' variant from~\cite{GuptaHRU11,GuptaRU12,HardtLM12}, in which the $k$ loss functions $\lossset = \set{\loss^1,\dots,\loss^k}$ are specified in advance by the analyst.  The offline variant contains the main novel ideas, although we will present our algorithm for the online case.

The algorithm receives as input a dataset $D \in \univ^{\rows}$ and a set of queries $\queryset$.  It will be useful to represent $D$ as a ``histogram'' over $\univ$, which is a vector indexed by $\univ$ where the $\row$-th entry is the probability that a random row of $D$ has type $\row$.  In this representation, a linear query $q$ can be written as $\langle q, D \rangle$.

The algorithm begins with a hypothesis dataset $D^1$, which represents an uneducated guess about $D$. It then produces a sequence of $T$ differentially private hypotheses $D^{1},\dots,D^{T}$ that are increasingly good approximations to $D$.  In each round $t = 1,\dots,T$, the algorithm will privately find the query $\query^{t} \in \queryset$ such that $D^{t}$ gives a maximally inaccurate answer.  That is, $|\langle q^{t}, D^{t} \rangle - \langle \query^{t}, D \rangle|$ is as large as possible.  Finding this query can be done privately using a standard application of the exponential mechanism~\cite{McSherryT07}.  The algorithm then generates $D^{t+1}$ using $D^t$ and $\query^t$ via the multiplicative weights update rule.

One can show that after a small number of rounds $T$, the hypothesis $D^{T}$ answers every query accurately. 
The key to the analysis is the following standard fact about the multiplicative-weights update rule: if one can find a vector $u^t$ such that $|\langle u^{t}, D^{t} \rangle -  \langle u^{t}, D \rangle|$ is large, then the distance between $D^{t+1}$ and $D$ decreases significantly.  Notice that this condition on $u^t$ is precisely that $u^t$ is a linear query for which $D^{t}$ is inaccurate.  Thus, when answering linear queries, we can simply take $u^t$ to be $\query^t$.  

In the case of CM queries, we can still use the exponential mechanism to find a loss function $\loss^{t} \in \lossset$ such that the minimizer of $\loss^{t}$ on $D^{t}$ is not a good minimizer of the loss on the true dataset $D$.  However, since CM queries are non linear, this information does not immediately give us a suitable vector $u^{t}$ for the multiplicative-weights update.  The key new step in our algorithm is a differentially private way to find a suitable vector $u^{t}$.  Specifically, we show how to take a query $\query_{\loss}$ such that $\query_{\loss}(\db^{t})$ is inaccurate for the true dataset $\db$, and a differentially private approximation to the correct answer $\query_{\loss}(\db)$, and use it to find a differentially private vector $u^{t}$ such that the error $|\langle u^{t}, D^{t} \rangle -  \langle u^{t}, D \rangle|$ is large.  As with linear queries, having such vectors is sufficient to argue accuracy of the algorithm.

Our approach is inspired by the work of Kasiviswanathan, Rudelson, and Smith~\cite{KasiviswanathanRS13} who prove lower bounds on the error required for answering certain CM queries.  Specifically, they use sufficiently accurate answers to non linear CM queries to extract linear constraints on the dataset, and these linear constraints can then be combined with linear reconstruction attacks to violate privacy.  For our results, we use the information that $D^t$ gives an inaccurate answer to a non linear CM query to find a linear query that $D^t$ also answers inaccurately.  To do so, we make use of the ``dual certificate'' style of argument from convex optimization.  That is, we derive and analyze the linear query using the first-order optimality conditions on the gradient of $\loss$.

\subsection{Connection to Generalization Error in Adaptive Data Analysis}
Very recently, Dwork et al.~\cite{DworkFHPRR15} and Hardt and Ullman~\cite{HardtU14} showed a connection between differential privacy and \emph{generalization error} in adaptive data analysis, in which the analyst asks an adaptively chosen sequence of queries.  By generalization error, we mean the difference between the answers to the queries on the dataset $\db$ and the answers to the queries on the unknown population from which $\db$ was drawn.  Dwork et al.~ showed that differentially private algorithms that have low error with respect to the dataset $\db$ also have low generalization error.  Surprisingly, using known differentially private algorithms for answer linear queries yields state-of-the-art bounds on the generalization error required to answer an interactive sequence of linear queries.  Bassily et al.~\cite{BassilySSU15} extended the connection between differential privacy and generalization error to the more general family of CM queries.  Plugging the results of this paper into their theorem yields state-of-the-art bounds on the generalization error required to answer adaptively chosen CM queries.

\section{Preliminaries}

\subsection{Datasets \ifnum\podscamera=1 \else,Histograms, \fi and Differential Privacy}
We define a \emph{dataset} $\mathcal{\db} \in \univ^{\rows}$ to be a vector of $\rows$ rows $\mathcal{\db} = (x_1,\dots,x_{\rows}) \in \univ^{\rows}$ from a \emph{data universe} $\univ$.  We say that two datasets $\mathcal{\db}, \mathcal{\db'} \in \univ^{\rows}$ are \emph{adjacent} if they differ on only a single row, and we denote this by $\mathcal{\db} \sim \mathcal{\db'}$.
\ifnum\podscamera=1
We now define \emph{differential privacy}~\cite{DworkMNS06}.
\begin{definition}
\else
\begin{definition} [Differential Privacy~\cite{DworkMNS06}]
\fi
\label{def:dp} An algorithm $\mech \from \univ^{\rows} \to \cR$ \ifnum\podscamera=1 satisfies \else is \fi \emph{$(\eps, \delta)$-differentially private} if for every two adjacent datasets $\mathcal{\db} \sim \mathcal{\db'}$ and every subset $S \subseteq \cR$,
$$
\prob{\mech(\mathcal{D}) \in S} \leq e^{\eps} \cdot \prob{\mech(\mathcal{D'}) \in S} + \delta.
$$
\end{definition}

In our algorithm and analysis it will be useful to represent a dataset by its \emph{histogram}.  In the histogram representation, the dataset $\mathcal{D}$ is viewed as a probability distribution over $\univ$.  We represent this probability distribution as a vector in $\db \in\R^{\univ}$ where for every $\row \in \univ$, $\db(\row) = \Prob{\row' \getsr \mathcal{\db}}{\row' = \row}.$ The condition that $\mathcal{D} \sim \mathcal{D'}$ implies that their histograms satisfy $\|D - D'\|_{1} \leq 1/\rows$.  In the technical sections of this work we will assume all datasets are represented as histograms.

\subsection{Convex Minimization (CM) Queries and Accuracy}
In this work we are interested in algorithms that answer \emph{convex minimization (CM) queries} on the dataset.  A CM query is defined by a convex \emph{loss function} $\loss \from \dom \times \univ \to \R$, where $\dom \subseteq \R^{\paramdimension}$ is a convex set.  The associated query $\query_{\loss} \from \univ^{*} \to \dom$ seeks to find $\param \in \dom$ that minimizes the expected loss.  Formally,
$$
\query_{\loss}(\db) = \argmin_{\param \in \dom} \Ex{\row \getsr \db}{\loss(\param; \row)} = \argmin_{\param 
\in \dom} \sum_{\row \in \univ} \db(\row) \cdot \loss(\param; \row)
$$

We will use $\lossset = \set{\loss_{1},\loss_{2},\dots}$ to denote a set of convex loss functions and $\queryset_{\lossset} = \set{\query_{\loss_{1}}, \query_{\loss_{2}},\dots}$ to denote the associated set of convex minimization queries.  
We will often want to think of $\loss$ as a function of $\param$, with $\row$ fixed.  To this end, we will write $\loss_{\row}(\param) = \loss(\param; \row)$.  We will also abuse notation and write $\loss(\param; \db) = \sum_{\row \in \univ} \db(\row) \cdot \loss(\param; \row)$ and $\loss_{\db}(\param) = \loss(\param; \db)$.

In order to define what it means to answer a CM query accurately, we define the following notion of \emph{error}, also known as ``excess empirical risk''.
\ifnum\podscamera=1
\begin{definition}
\else
\begin{definition}[Error of an Answer]
\fi
\label{def:answererror}
For a loss function $\loss \from \dom \times \univ \to \R$, database $\db \in \univ^*$, and answer $\hat{\param} \in \dom$, we define \emph{the error of $\hat{\param}$ on $\loss$ with respect to $\db$} to be
$$
\err_{\loss}(\db, \hat{\param}) = \loss(\hat{\param}; \db) - \min_{\param \in \dom} \loss(\param; \db).
$$
\end{definition}
It will also be useful in describing an analyzing out algorithm to define the notion of \emph{error of a database} as follows.
\ifnum\podscamera=1
\begin{definition}
\else
\begin{definition}[Error of a Database]
\fi
\label{def:dberror}
For a loss function $\loss \from \dom \times \univ \to \R$, database $\db \in \univ^*$, and another database $\db' \in \univ^*$, we define \emph{the error of $\db'$ on $\loss$ with respect to $\db$} to be
$$
\err_{\loss}(\db, \db') =  \loss_{\db}\left(\argmin_{\param' \in \dom} \loss_{\db'}(\param')\right) - \min_{\param \in \dom} \loss_{\db}(\param).
$$
\end{definition}

%%%%%%%
%%%%%%%
\begin{comment}
\begin{definition} [Accuracy for a Single Query] \label{def:oneacc}
Let $\loss$ be a loss function and $\query_{\loss}$ be the associated CM query.  An answer $\hat{\param}$ is \emph{$\acc$-accurate for $\query_{\loss}$ on $\db$} if
$
\loss_{\db}(\hat{\param}) \leq \min_{\param \in \dom} \loss_{\db}(\hat{\param}) + \alpha.
$

An algorithm $\mech$ is \emph{$(\acc, \beta)$-accurate for $\query_{\loss}$ on datasets of size $\rows$} if for any $\db \in \univ^{\rows}$, with probability at least $1-\beta$, $\mech(\db)$ outputs $\hat{\param}$ that is $\acc$-accurate for $\query_{\loss}$ on $\db$.
\end{definition}

\begin{definition} [Accuracy for a Family of Queries] \label{def:manyacc}
Let $\lossset$ be a family of loss functions and $\queryset_{\lossset}$ be the associated family of CM queries.  A vector of answers $(\hat{\param}_{\loss})_{\loss \in \lossset}$ is \emph{$\acc$-accurate for $\queryset_{\lossset}$ on $\db$} if for every $\loss \in \lossset$, $\hat{\param}_{\loss}$ is $\alpha$-accurate for $\loss$ on $\db$.

An algorithm $\mech$ is \emph{$(\acc, \beta)$-accurate for $\queryset_{\lossset}$ on datasets of size $\rows$} if for any $\db \in \univ^{\rows}$, with probability at least $1-\beta$, $\mech(\db)$ outputs $(\hat{\param}_{\loss})_{\loss \in \lossset}$ that is $\alpha$-accurate for $\queryset_{\lossset}$ on $\db$.
\end{definition}
\end{comment}
%%%%%%%
%%%%%%%

We now define what it means for an algorithm $\cA$ to be \emph{accurate} for answering a sequence of CM queries from a family $\lossset$.  We do so by means of a game between $\cA$ and an adversary $\cB$, defined in Figure~\ref{fig:sampleaccgame}.

\begin{figure}[ht!]
\begin{framed}
\begin{algorithmic}
\STATE{$\cB$ chooses $\db \in  \univ^n$.}
\STATE{For $j = 1,\dots,|\lossset|$}
\INDSTATE[1]{$\cB$ outputs a loss function $\loss^j \in \lossset$.}
\INDSTATE[1]{$\cA(\db, \loss^j)$ outputs $\hat{\param}^j$.}
\INDSTATE[1]{(As $\cB$ and $\cA$ are stateful, $\loss^j$ and $\hat{\param}^j$ may depend on the history $\loss^1,\hat{\param}^1,\dots,\loss^{j-1},\hat{\param}^{j-1}$.)}
%\INDSTATE[0]{Let $\lossset = \set{\loss_1,\dots,\loss_{|\lossset|}}$.}
\end{algorithmic}
\end{framed}
\vspace{-6mm}
\caption{The Sample Accuracy Game $\mathsf{Acc}_{n, k, \lossset}[\cA, \cB]$\label{fig:sampleaccgame}}
\end{figure}

\ifnum\podscamera=1
\begin{definition}
\else
\begin{definition}[Accuracy] 
\fi
\label{def:accuracy}
Let $\lossset$ be a set of convex loss functions and $\queryset_{\lossset}$ be the associated set of CM queries.  Let $0 < \alpha, \beta \leq 1$ and $k, n \in \N$ be parameters.  We say that an algorithm $\cA$ is \emph{$(\alpha, \beta)$-accurate for answering $k$ CM queries from $\queryset_{\lossset}$ given a database of size $n$} if for every adversary $\cB$,
$$
\Prob{\mathsf{Acc}_{n, k, \lossset}}{\max_{j = 1,\dots,k} \err_{\loss^j}(\db, \hat{\param}^j) \leq \alpha} \geq 1 - \beta.
$$
\end{definition}

\section{Online Private Multiplicative \ifnum\podscamera=1 \\ \else \fi Weights for CM Queries}
In this section we present and analyze a differentially private algorithm that answers any family of CM queries provided black-box access to a differentially private algorithm that answers any single CM query from the family. 

\subsection{The Online Sparse Vector Algorithm} \label{sec:onlinesv}
Just like when using private multiplicative weights to answer linear queries, a key ingredient in our algorithm is the \emph{online sparse vector algorithm}.  At a high level, the online sparse vector algorithm takes a database $\db \in \univ^n$ and a sequence of queries $q_1,\dots,q_k$, but it provides only a very weak accuracy guarantee.  Each query is answered with a single bit in $\set{\top, \bot}$.  For a given query $q$ and some threshold $\alpha$, if $q(\db) \geq \alpha$ then the algorithm answering $\top$ and if $q(\db) \leq \alpha/2$ it answers $\bot$.  If the answer is in $(\alpha/2, \alpha)$ any answer is allowed.  The key feature of the online sparse vector algorithm is that the size of the dataset $n$ only needs to be proportional to $\sqrt{T} \cdot \log k$, where $T$ is the number of queries whose answer is above the threshold.  \ifnum\podscamera=1 Whereas\else In contrast\fi, approximately answering every query requires $n$ to grow like $\sqrt{k}$.

\newcommand{\cSV}{\mathcal{SV}}
To maintain brevity, and since the algorithm is standard (see~\cite{DworkR14} for a textbook treatment), we will not specify the algorithm.  Instead we will define its properties as a black box.  We define the guarantees of the sparse vector algorithm via the following game between the online sparse vector algorithm $\cSV$ and an adversary $\cB$.
\begin{figure}[h!]
\begin{framed}
\begin{algorithmic}
\INDSTATE[0]{$\cB$ chooses a dataset $\db \in \univ^n$.}
\INDSTATE[0]{For $j = 1,\dots,k$:}
\INDSTATE[1]{$\cB$ outputs a $(3S/n)$-sensitive query $q^j$}
\INDSTATE[1]{(The query $q^j$ may depend on the previous queries and answers $q^1,a^1,\dots,q^{j-1},a^{j-1}$.)}
\INDSTATE[1]{$\cSV$ returns an answer $a^j \in \set{\top, \bot}$.}
\end{algorithmic}
\end{framed}
\vspace{-6mm}
\caption{$\mathsf{ThresholdGame}_{n, T, k, \alpha}[\cSV, \cB]$}
\label{fig:sv}
\end{figure}

The requirement that $\cB$ outputs a $(3/n)$-sensitive query means that $q$ satisfies $|q(\db) - q(\db')| \leq 3S/n$ for every pair of neighboring databases $\db \sim \db' \in \univ^n$.  The choice of $(3S/n)$ can be replaced with any parameter $\Delta$, but we fix it to $3S/n$ to cut down on notation, since we'll use that choice in the next section.

\begin{theorem} \label{thm:onlinesvacc}
There is an algorithm $\cSV = \cSV(T, k, \alpha, \eps, \delta)$ such that for every $T, k \in \N$ and $\alpha, \eps, \delta > 0$, the following three conditions hold.
\begin{enumerate}
\item $\cSV$ is $(\eps, \delta)$-differentially private.
\item $\cSV$ halts if $T$ queries are answered with $\top$.
\item If
$$
n \geq \frac{256 \cdot S \cdot \sqrt{T \cdot \log(2/\delta)} \cdot \log (4k/\beta)}{\eps \alpha},
$$
then
\ifnum\podscamera=1
$$
\Prob{\mathsf{ThresholdGame}_{n, T, k, \tau}[\cSV, \cB]}{\mathit{WIN}} \geq 1-\beta,
$$
where
$$
\mathit{WIN} = \left\{\forall j \in [k],\; {q^j(\db) \geq \alpha \Longrightarrow a^j = \top \atop q^j(\db) \leq \alpha/2 \Longrightarrow a^j = \bot}\right\}
$$
\else
$$
\Prob{\mathsf{ThresholdGame}_{n, T, k, \tau}[\cSV, \cB]}{\forall j \in [k],\; {q^j(\db) \geq \alpha \Longrightarrow a^j = \top \atop q^j(\db) \leq \alpha/2 \Longrightarrow a^j = \bot}} \geq 1-\beta.
$$
\fi
\end{enumerate}
\end{theorem}

\subsection{The Algorithm}
We are now ready to describe our algorithm for answering exponentially many convex minimization queries from some family $\lossset = \set{\loss \from \dom \times \univ \to \R}$.  Assume every $\loss \in \lossset$ satisfies the scaling condition
$$
\max_{\row \in \univ, \param, \param' \in \dom} \left|\left\langle \param - \param', \grad \loss_{\row}(\param) \right\rangle\right| \leq S.
$$

The algorithm is defined in Figure~\ref{fig:onlinepmw}.  Note that in the algorithm there are two sequences of queries that it will be useful to distinguish.  The first is the set of queries actually issued by the analyst, which are index by the letter $j$ and are $\loss^1,\dots,\loss^k$.  There is also the \emph{subsequence} of queries such that $a^j = \top$ and lead to updates.  We use the letter $t$ to index these queries, which are $\loss^1,\dots, \loss^T$ (there cannot be more than $T$ such queries, since $\cSV$ would halt, though there may be fewer).  Sometimes it will be useful to consider only the subsequence of queries that are used for updates, which is why we use a separate index for this sequence.

%%%%%
\ifnum\podscamera=1
\begin{figure}[h!]
\begin{framed}
\begin{algorithmic}
\INDSTATE[0]{Input and parameters: A dataset $\db \in \univ^{\rows}$, parameters $\eps, \delta, \alpha, \beta, S, k > 0$, and oracle access to $\mech'$, an $(\eps_0, \delta_0)$-differentially private algorithm that is $(\alpha_0, \beta_0)$-accurate for one convex minimization query in $\lossset$ on datasets of size $\rows'$, for parameters $\eps_0, \delta_0, \alpha_0, \beta_0$.}
	\INDSTATE[0]{$$ T = \frac{64S^2 \log |\univ|}{\alpha^2}  \qquad \eta = \sqrt{\frac{\log |\univ|}{T}} $$}
	\INDSTATE[0]{$$\eps_0 = \frac{\eps}{\sqrt{8T \log(4/\delta)}} \qquad \delta_0 = \frac{\delta}{4T}$$}
	\INDSTATE[0]{$$ \alpha_0 = \frac{\alpha}{4} \qquad \beta_0 = \frac{\beta}{2T}$$}
\INDSTATE[0]{Let $\cSV = \cSV(T, k, \alpha, \eps/2, \delta/2)$ be the online sparse vector algorithm (Section~\ref{sec:onlinesv}).}
\INDSTATE[0]{Let $t = 1$.  Let $\dbt^{t} \in \R^{\univ}$ be the uniform histogram over $\univ$.}
\INDSTATE[0]{For $j = 1,\dots,k$:}
\INDSTATE[1]{Receive loss function $\loss = \losst^j$.}
\INDSTATE[1]{Let $q^j$ be the $(3/n)$-sensitive query $$q^j(\db) = \err_{\loss} (\db, \dbt^{t}).$$}
\INDSTATE[1]{Run $\cSV$ on $q^j$, to obtain an answer $a^j \in \set{\top, \bot}$.}
\INDSTATE[1]{(If $\cSV$ halts, then halt.)}
\INDSTATE[1]{If $a^j = \bot$:}
\INDSTATE[2]{Output the answer $\hat{\param}^j = \argmin_{\theta \in \Theta} \loss(\theta; \dbt^{t}).$}
\INDSTATE[1]{Else if $a^j = \top$:}
\INDSTATE[2]{Let $\loss^t = \loss$.}
\INDSTATE[2]{Let $\privparamt^{t} \getsr \mech'(\db, \loss^t)$ be a private estimate of the}
\INDSTATE[2]{minimizer of $\loss^t$ on $\db$.}
\INDSTATE[2]{Output the answer $\hat{\param}^j = \privparamt^{t}$.}
\INDSTATE[2]{Update $\dbt^{t}$:}
\INDSTATE[3]{Let $\paramt^{t} = \argmin_{\param \in \dom}  \loss(\theta; \dbt^{t})$}
\INDSTATE[3]{and let $u^t \in [-\scale,\scale]^{\univ}$ be the vector
$$
u^t(\row) = \left\langle \privparamt^{t} - \paramt^{t}, \grad \losst^{t}_{\row}(\paramt^{t})\right\rangle
$$}
\vspace{-4mm}
\INDSTATE[3]{Let $\dbt^{t+1}(\row) \propto \exp(\eta \cdot u^{t}(\row)) \cdot \dbt^{t}(\row)$}
\INDSTATE[3]{Let $t = t+1$.}
\INDSTATE[3]{(Note that $t \leq T$, or $\cSV$ would have halted.)}
\end{algorithmic}
\end{framed}
\vspace{-6mm}
\caption{Online Private Multiplicative Weights for CM Queries}
\label{fig:onlinepmw}
\end{figure}
%%%%%
\else
%%%%%
\begin{figure}[h!]
\begin{framed}
\begin{algorithmic}
\INDSTATE[0]{Input and parameters: A dataset $\db \in \univ^{\rows}$, parameters $\eps, \delta, \alpha, \beta, S, k > 0$, and oracle access to $\mech'$, an $(\eps_0, \delta_0)$-differentially private algorithm that is $(\alpha_0, \beta_0)$-accurate for one convex minimization query in $\lossset$ on datasets of size $\rows'$, for parameters $\eps_0, \delta_0, \alpha_0, \beta_0$.}
	\INDSTATE[0]{$$ T = \frac{64S^2 \log |\univ|}{\alpha^2}  \qquad \eta = \sqrt{\frac{\log |\univ|}{T}} $$}
	\INDSTATE[0]{$$\eps_0 = \frac{\eps}{\sqrt{8T \log(4/\delta)}} \qquad \delta_0 = \frac{\delta}{4T} \qquad \alpha_0 = \frac{\alpha}{4} \qquad \beta_0 = \frac{\beta}{2T}$$}
\INDSTATE[0]{Let $\cSV = \cSV(T, k, \alpha, \eps/2, \delta/2)$ be the online sparse vector algorithm (Section~\ref{sec:onlinesv}).}
\INDSTATE[0]{Let $t = 1$.  Let $\dbt^{t} \in \R^{\univ}$ be the uniform histogram over $\univ$.}
\INDSTATE[0]{For $j = 1,\dots,k$:}
\INDSTATE[1]{Receive loss function $\loss = \losst^j \in \lossset$.}
\INDSTATE[1]{Let $q^j$ be the $(3/n)$-sensitive query $q^j(\db) = \err_{\loss} (\db, \dbt^{t}).$}
\INDSTATE[1]{Run $\cSV$ on $q^j$, to obtain an answer $a^j \in \set{\top, \bot}$. (If $\cSV$ halts, then halt.)}
\INDSTATE[1]{If $a^j = \bot$:}
\INDSTATE[2]{Output the answer $\hat{\param}^j = \argmin_{\theta \in \Theta} \loss(\theta; \dbt^{t}).$}
\INDSTATE[1]{Else if $a^j = \top$:}
\INDSTATE[2]{Let $\loss^t = \loss$.}
\INDSTATE[2]{Let $\privparamt^{t} \getsr \mech'(\db, \loss^t)$ be a private estimate of the minimizer of $\loss^t$ on $\db$.}
\INDSTATE[2]{Output the answer $\hat{\param}^j = \privparamt^{t}$.}
\INDSTATE[2]{Update $\dbt^{t}$:}
\INDSTATE[3]{Let $\paramt^{t} = \argmin_{\param \in \dom}  \loss(\theta; \dbt^{t})$ and let $u^t \in [-\scale,\scale]^{\univ}$ be the vector
$$
u^t(\row) = \left\langle \privparamt^{t} - \paramt^{t}, \grad \losst^{t}_{\row}(\paramt^{t})\right\rangle
$$}
\vspace{-4mm}
\INDSTATE[3]{Let $\dbt^{t+1}(\row) \propto \exp(\eta \cdot u^{t}(\row)) \cdot \dbt^{t}(\row)$}
\INDSTATE[3]{Let $t = t+1$.  (Note that $t \leq T$, otherwise $\cSV$ would have halted.)}
\end{algorithmic}
\end{framed}
\vspace{-6mm}
\caption{Online Private Multiplicative Weights for CM Queries}
\label{fig:onlinepmw}
\end{figure}
\fi
%%%%%

\subsection{Accuracy Analysis} \label{sec:accuracy}

In this section, we prove that our algorithm is accurate for any family of CM queries $\lossset$, provided that the oracle $\mech'$ is accurate for any single CM query from $\lossset$.  As with previous variants of private multiplicative weights~\cite{HardtR10, GuptaHRU11, GuptaRU12, HardtLM12}, we will derive the accuracy guarantee using the well known ``bounded regret'' property of the multiplicative weights update rule, combined with the utility guarantees of the online sparse vector algorithm.

To start the analysis we will assume that two conditions are satisfied.  
%First, we use the guarantee that $\cSV$ will halt if it answers $T$ queries with $\top$
%\begin{equation} \label{eq:svhalts}
%\textrm{$\cSV$ halts if $T$ queries are answered with $\top$}
%\end{equation}
First, we assume that $\cSV$ answered accurately---formally, we assume that
\begin{equation} \label{eq:svaccurate}
\forall j \in [k],\; {\err_{\loss^j}(\db; \dbt^{t}) \geq \alpha \Longrightarrow a^j = \top \atop \err_{\loss^j}(\db; \dbt^{t}) \leq \alpha/2 \Longrightarrow a^j = \bot}
\end{equation}
where $\dbt^{t}$ is the current dataset $\dbt^{t}$ that is in use at the time the loss function $\loss^j$ is considered.  By the accuracy of the online sparse vector algorithm $\cSV$ (Theorem~\ref{thm:onlinesvacc}), the event~\eqref{eq:svaccurate} holds with probability at least $1-\beta/2$ as long as $n$ is sufficiently large.  %Our analysis will show that indeed $\cSV$ will be asked at most $T$ queries such that $q^j(\db) \geq \alpha$.

Second, we will assume that every time $a^j = \top$ and $\mech'(\db, \loss^j)$ is called, it returns an accurate answer---formally,
\begin{equation} \label{eq:oneshotaccurate}
\forall \textrm{$j$ such that $a^j = \top$},\; \err_{\loss^j}(\db, \privparamt^{t}) \leq \alpha_0.
\end{equation}
Since $\mech'$ is assumed to be $(\alpha_0, \beta_0)$ accurate for one query provided that $n \geq n'$, and $\mech'$ is called at most $T$ times, we can conclude that the event~\eqref{eq:oneshotaccurate} holds with probability at least $1-\beta/2$.  The following claim is immediate.

\begin{claim} \label{clm:conditiononsuccess}
If
$$
n \geq \max\left\{ n', \frac{512 \cdot \sqrt{T \cdot \log(4/\delta)} \cdot \log (8k/\beta)}{\eps \alpha} \right\},
$$
then with probability at least $1-\beta$, the events~\eqref{eq:svaccurate} and~\eqref{eq:oneshotaccurate} both hold.
\end{claim}

Thus, we are justified proving that the online private multiplicative weights algorithm is accurate conditioned on~\eqref{eq:svaccurate} and~\eqref{eq:oneshotaccurate}.  We start by observing that the algorithm can only fail to be accurate if it halts before the entire sequence of $k$ queries has been asked (because $t = T$ updates have been performed and $\cSV$ halted).
\begin{claim} \label{clm:doesnotterminate}
Assume that the algorithm does not terminate before answering $k$ queries, and that~\eqref{eq:svaccurate} and~\eqref{eq:oneshotaccurate} both hold.  Then the algorithm answers every query with error at most $\alpha$.  That is,
$$
\forall j \in [k],\; \err_{\loss^j}(\db, \hat{\param}^j) \leq \alpha.
$$
\end{claim}
\begin{proof}[Proof of Claim~\ref{clm:doesnotterminate}]
If the algorithm has not terminated, then each query $\loss^j$ is answered in one of two ways.  If $a^j = \bot$, then we answer with $\hat{\param}^j = \argmin_{\theta \in \Theta} \loss(\theta; \dbt^{t}).$  In this case, since~\eqref{eq:svaccurate} holds,  and $a^j = \bot$, we have $\err_{\loss^j}(\db, \hat{\param}^j) \leq \alpha$.  But, by definition, $\err_{\loss^j}(\db, \dbt^{t}) = \err_{\loss^j}(\db, \hat{\param}^j)$.  So the algorithm answers accurately in the case where $a^j = \bot$.
$\err_{\loss^j}(\db, \hat{\param}^j) \leq \alpha$

If $a^j = \top$, then we answer with $\hat{\param}^j = \privparamt^{j} = \mech'(\db, \loss^{j})$.  Since~\eqref{eq:oneshotaccurate} holds, we have $$\err_{\loss^j}(\db, \hat{\param}^j) = \err_{\loss^j}(\db, \privparamt^{t}) \leq \alpha_0 \leq \alpha,$$ as desired.
\end{proof}

To complete the proof, it suffices to show that the algorithm does not terminate early.  Here is where we rely on the ``bounded regret'' property of the multiplicative weights update rule.
\begin{lemma} \label{lem:regretboundonline} [See e.g.~\cite{AroraHK12}]
For every sequence $u^{1}, \dots, u^{T} \in [-\scale,\scale]^{\univ}$,
$$
\frac{1}{T} \sum_{t=1}^{T} \left\langle u^{t}, \dbt^{t} - \db \right\rangle \leq 2\scale \sqrt{\frac{\log |\univ|}{T}}
$$
\end{lemma}

Recall that the algorithm only terminates early if there are $T$ queries $\loss^j$ such that $a^j = \top$, and by~\eqref{eq:svaccurate}, $a^j = \top$ only if the error of $\dbt^t$ on $\loss^j$ is at least $\alpha/2$.  Thus, in light of the preceding lemma, we would like to show that if $\dbt^t$ has error $\alpha/2$ for a query $\loss$, then $\langle u^{t}, \dbt^{t} - \db \rangle$ is also large, say $\alpha/4$.  If we can show such a statement, then by our choice of $T$, it will be impossible to perform a sequence of $T$ updates, and thus the algorithm will not terminate early.

The key lemma, and the main novelty in our analysis, is to relate $\langle u^{t}, \dbt^{t} - \db \rangle$ to the error of $\dbt^t$ on a query $\loss^j$.  We show that $\langle u^{t}, \dbt^{t} - \db \rangle$ is at least the additional loss incurred by $\paramt^{t}$ over that of $\privparamt^{t}$. 
\begin{claim} \label{clm:potentialdroponline}
For every $t = 1,\dots,T$,
$$
\left\langle u^t, \dbt^{t} - \db \right\rangle \geq \losst^{t}_{\db}(\paramt^{t}) - \losst^{t}_{\db}(\privparamt^{t})
$$
\end{claim}
Recall that $\privparamt^{t}$ is an approximation to the optimal solution for $\losst^{t}_{\db}$, whereas $\paramt^{t}$ has large error with respect to $\db$.  Thus we expect the right hand side of the expression to be positive and large.
\ifnum\podscamera=1
\begin{proof}
\else
\begin{proof} [Proof of Claim~\ref{clm:potentialdroponline}]
\fi
Recall that we chose
$$
\paramt^{t} = \argmin_{\param \in \dom} \losst^{t}_{\dbt^{t}}(\param).
$$
By the first-order optimality condition, and the fact that $\privparamt^{t}, \paramt^{t} \in \dom$ for a convex set $\dom$, the directional derivative of $\losst^{t}_{\dbt^{t}}$ at $\paramt^{t}$ in the direction of $\privparamt^{t} -  \paramt^{t}$ will be positive.  So we have
\begin{align}
0 
\leq{} \left\langle \privparamt^{t} - \paramt^{t} , \grad \losst^{t}_{\dbt^{t}}(\paramt^{t}) \right\rangle
={} &\sum_{\row \in \univ} \dbt^{t}(\row) \cdot \left\langle \privparamt^{t} - \paramt^{t} , \grad \losst^{t}_{\row}(\paramt^{t}) \right\rangle \notag \\
={} &\left\langle u^t, \dbt^{t} \right\rangle \label{eq:optonline}.
\end{align}
The first equality uses linearity of the gradient and the definition
$
\losst^{t}_{\dbt^{t}}(\cdot) = \sum_{\row \in \univ} \dbt^{t}(\row) \cdot \losst^{t}_{\row}(\cdot)
$

Similarly, we can look at the directional derivative of $\losst^{t}_{\db}$ again taken at $\paramt^{t}$ and in the direction of $\privparamt^{t} -  \paramt^{t}$.
\begin{align}
\left\langle \privparamt^{t} - \paramt^{t} , \grad \losst^{t}_{\db}(\paramt^{t}) \right\rangle 
={} &\sum_{\row \in \univ} \db(\row) \cdot \left\langle \privparamt^{t} - \paramt^{t} , \grad \losst^{t}_{\row}(\paramt^{t}) \right\rangle \notag \\
={} &\left\langle u^{t}, \db \right\rangle \label{eq:suboptonline}.
\end{align}
If $\paramt^{t}$ is far from optimal for the input dataset $\db$, then moving in the direction of $\privparamt^{t} - \paramt^{t}$ should significantly decrease the loss.  Thus, since $\loss$ is convex, this directional derivative must be significantly negative.  Specifically, since $\losst^{t}_{\db}$ is convex, $\losst^{t}_{\db}$ lies above all of its tangent lines.  Thus,  
\begin{align*}
\losst^{t}_{\db}(\privparamt^{t}) 
\geq \losst^{t}_{\db}(\paramt^{t}) + \left\langle \privparamt^{t} - \paramt^{t}, \grad \losst^{t}_{\db}(\paramt^{t}) \right\rangle ={} \losst^{t}_{\db}(\paramt^{t}) + \left\langle u^{t}, \db \right\rangle.
\end{align*}
where the equality is from~\eqref{eq:suboptonline}
Rearranging terms, we have
\begin{equation}
- \left\langle u^{t},  \db \right\rangle 
\geq \losst^{t}_{\db}(\paramt^{t}) - \losst^{t}_{\db}(\privparamt^{t}). \label{eq:lossonline}
\end{equation}
Combining \eqref{eq:optonline} and \eqref{eq:lossonline}, we have
\begin{align*}
&\left\langle u^t,  \dbt^{t} - \db \right\rangle \geq \losst^{t}_{\db}(\paramt^{t}) - \losst^{t}_{\db}(\privparamt^{t}),
\end{align*}
which completes the proof.
\end{proof}

Using Claim~\ref{clm:potentialdroponline}, and the guarantees~\eqref{eq:svaccurate} and~\eqref{eq:oneshotaccurate}, we can now lower bound $\langle u^{t}, \dbt^{t} - \db \rangle$.
\begin{claim} \label{clm:finalpotentialdroponline}
For every $t = 1,\dots,T$, if the algorithm has not terminated, and~\eqref{eq:svaccurate} and~\eqref{eq:oneshotaccurate} both hold, then
$$\langle u^{t}, \dbt^{t} - \db \rangle > \alpha/4.$$
\end{claim}
\ifnum\podscamera=1
\begin{proof}
\else
\begin{proof}[Proof of Claim~\ref{clm:finalpotentialdroponline}]
\fi
Our goal is to lower bound $\langle u^t, \dbt^{t} - \db \rangle$ by the quantity
$
\err_{\losst^{t}}(\db, \dbt^{t}) = \losst^{t}_{\db}(\paramt^{t}) - \min_{\param \in \dom} \losst^{t}_{\db}(\param).
$
This condition is almost implied by Claim~\ref{clm:potentialdroponline}, except with $\losst^{t}_{\db}(\privparamt^{t})$ in place of the minimum.  In the next claim, we extend the previous claim to handle an approximate minimizer.

However, by~\eqref{eq:oneshotaccurate}, $\privparamt^{t} = \mech'(\db, \losst^{t})$ is an approximate minimizer.  That is, 
\begin{equation} \label{eq:approxminonline}
\losst^{t}_{\db}(\privparamt^{t}) \leq \min_{\param \in \dom} \losst^{t}_{\db}(\param) + \alpha_0.
\end{equation}
Combining Claim~\ref{clm:potentialdroponline} with~\eqref{eq:approxminonline} we conclude that if $\rows \geq n'$, then for every $t = 1,\dots,T$, with probability at least $1-\beta_0$,
\ifnum\podscamera=1
\begin{align}
\left\langle u^t, \dbt^{t} - \db \right\rangle 
\geq{} &\losst^{t}_{\db}(\paramt^{t}) - \left(\min_{\param \in \dom} \losst^{t}_{\db}(\param) + \alpha_0\right) \notag \\ 
={} &\err_{\losst^{t}}(\db, \dbt^{t}) - \alpha_0 \label{eq:potentialdrop2online}
\end{align}
\else
\begin{equation} \label{eq:potentialdrop2online}
\left\langle u^t, \dbt^{t} - \db \right\rangle \geq \losst^{t}_{\db}(\paramt^{t}) - \left(\min_{\param \in \dom} \losst^{t}_{\db}(\param) + \alpha_0\right) = \err_{\losst^{t}}(\db, \dbt^{t}) - \alpha_0
\end{equation}
\fi

Given~\eqref{eq:potentialdrop2online} we would like to show that $\err_{\losst^{t}}(\db, \dbt^{t})$ is large.  But, by~\eqref{eq:svaccurate}, we would only do an update if $\err_{\loss^t}(\db, \dbt^{t}) > \alpha/2$.  Therefore we must have
\ifnum\podscamera=1
\begin{align*}\left\langle u^t, \dbt^{t} - \db \right\rangle \geq \losst^{t}_{\db}(\paramt^{t}) - \left(\min_{\param \in \dom} \losst^{t}_{\db}(\param) + \alpha_0\right) >{} &\alpha/2 - \alpha_0 \\ ={} &\alpha/4,
\end{align*}
\else
\begin{equation*}\left\langle u^t, \dbt^{t} - \db \right\rangle \geq \losst^{t}_{\db}(\paramt^{t}) - \left(\min_{\param \in \dom} \losst^{t}_{\db}(\param) + \alpha_0\right) > \alpha/2 - \alpha_0 = \alpha/4,
\end{equation*}
\fi
as desired.
\end{proof}

We are now ready to show that the online private multiplicative weights algorithm does not terminate early.
\begin{claim} \label{clm:noearlytermination}
If~\eqref{eq:svaccurate} and~\eqref{eq:oneshotaccurate} both hold, then the algorithm does not terminate before answering $k$ queries.
\end{claim}
\ifnum\podscamera=1
\begin{proof}
\else
\begin{proof}[Proof of Claim~\ref{clm:noearlytermination}]
\fi
Assume for the sake of contradiction that the algorithm does terminate early because of the condition $t = T$.  Then, by Claim~\ref{clm:finalpotentialdroponline}, there is a sequence of $T$ queries such that for every query
$$
\langle u^{t}, \dbt^{t} - \db \rangle \geq \alpha/4.
$$
Then, using the bounded-regret property of multiplicative weights (Lemma~\ref{lem:regretboundonline}), we must have
\begin{align*}
\alpha/ 4
<{} &\frac{1}{T} \sum_{t=1}^{T}  \left\langle u^t, \dbt^{t} - \db \right\rangle \\
\leq{} &2S \sqrt{\frac{\log |\univ|}{T}} \tag{Lemma~\ref{lem:regretboundonline}} \\
\leq{} &\alpha/4,
\end{align*}
which is a contradiction.
\end{proof}

The analysis of this section immediately implies the following theorem
\begin{theorem} \label{thm:putittogetheronline}
The online private multiplicative weights algorithm is $(\alpha, \beta)$-accurate for answering $k$ CM queries from $\queryset_{\lossset}$ given a dataset of size $n$
for
$$
n = \max\left\{ n', \frac{4096 \cdot S^2 \cdot \sqrt{\log |\univ| \cdot \log(4/\delta)} \cdot \log (8k/\beta)}{\eps \alpha^2} \right\}.
$$
\end{theorem}

\subsection{Privacy Analysis}
In this section we show that our algorithm (Figure~\ref{fig:onlinepmw}) is differentially private.  Privacy will follow rather easily from privacy of the online sparse vector algorithm, privacy of $\mech'$, and well known composition properties of differential privacy.

\begin{theorem} \label{thm:dp}
If $\mech'$ is $(\eps_0, \delta_0)$-differentially private, for $\eps_0, \delta_0$ as stated, then the algorithm in Figure~\ref{fig:onlinepmw} is $(\eps,\delta)$-differentially private.
\end{theorem}

\subsubsection{Composition of Differential Privacy}

Before proceeding to the privacy analysis of our algorithm, we recall the composition properties of differential privacy.

A well-known fact about differential privacy is that the parameters $\eps, \delta$ degrade gracefully under composition.  Specifically, we will make use of the strong composition theorem due to Dwork, Rothblum, and Vadhan~\cite{DworkRV10}.  Formally, we say that an algorithm $\mech$ is a $T$-fold adaptive composition of $(\eps_0, \delta_0)$-differentially private algorithms if $\mech$ can be expressed as an instance of the following game for some adversary $\cB$:
\begin{figure}[ht]
\begin{framed}
\begin{algorithmic}
\STATE{Let $\db$ be a database, let $\cB$ be an adversary, $T$ be a parameter}
\STATE{For $t = 1,\dots,T$}
\INDSTATE{$\cB(z_{1},\dots,z_{t-1})$ outputs an $(\eps_0,\delta_0)$-DP $\mech_{t}$}
\INDSTATE{Let $z_{t} = \mech_{t}(\db)$}
\STATE{Output $z_1,\dots,z_{T}$}
\end{algorithmic}
\end{framed}
\vspace{-6mm}
\caption{$T$-Fold Adaptive Composition}
\label{fig:composition}
\end{figure}

\begin{theorem}[\cite{DworkRV10}] \label{thm:composition}
For every $T \in \N$ and $0 \leq \eps_0, \delta_0, \delta' \leq 1/2$, if $\mech$ is a $T$-fold adaptive composition of \ifnum\podscamera=1 algorithms that are $(\eps_0, \delta_0)$-differentially private\else $(\eps_0, \delta_0)$-differentially private algorithms\fi, then $\mech$ is $(\eps,\delta' + T\delta_0)$-differentially private for
$$
\eps = \sqrt{2T\log(1/\delta')} \cdot \eps_0 + 2 T \cdot \eps_0^2.
$$

In particular, if $\mech$ is a $T$-fold adaptive composition of $(\eps_0,\delta_0)$-differentially private algorithms, where
$$
\eps_0 = \frac{\eps}{\sqrt{8 T \log(2/\delta)}} \qquad \delta_0 = \frac{\delta}{2T},
$$
then $\mech$ is $(\eps, \delta)$-differentially private.
\end{theorem}

\subsubsection{Proof of Theorem~\ref{thm:dp}}

There are only two places where the algorithm uses the private dataset $\db$: (1) when using the online sparse vector algorithm to answer the queries $q^j = \err_{\loss^j}(\db, \dbt^{t})$, and (2) when using $\mech'$ to obtain a private approximation to the minimizer of some loss function $\loss^t$.  First, we will show that the online sparse vector algorithm is $(\eps/2, \delta/2)$-differentially private.  This claim will follow immediately from Theorem~\ref{thm:onlinesvacc} provided that the queries $q^j$ are indeed $(3S/n)$-sensitive.  To show this, first, observe that if $\loss \from \dom \times \univ \to \R$ satisfies
$$
\max_{\row \in \univ, \param, \param' \in \dom} \left|\left\langle \param - \param', \grad \loss_{\row}(\param) \right\rangle\right| \leq S,
$$
then for every $\row \in \univ$, there exists $b_{\row} \in \R$ such that for every $\param \in \dom$, $\loss(\param, \row) \in [b_{\row}, S]$.  That is, for every $\row$, there is some interval of width $S$ that bounds the loss $\loss(\param, \row)$.  With this information we can bound the sensitivity of the error function in the following way:  Fix any $\loss \in \lossset$.  Let $\overline{\ell}(\param, \row) = \ell(\param, \row) - b_{\row}$.  Let $\overline{\param} = \argmin_{\param \in \dom} \loss_{\dbt^{t}}(\param)$.
\begin{align*}
&\max_{\db, \db' \in \univ^{\rows}} \left| \err_{\loss}(\db, \dbt^{t}) - \err_{\loss}(\db', \dbt^{t}) \right| \\
={} &\max_{\db, \db' \in \univ^{\rows}} \left| \left(\loss_{\db}(\overline{\param}) - \min_{\param \in \dom} \loss_{\db}(\param)\right) - \left(\loss_{\db'}(\overline{\param}) - \min_{\param \in \dom} \loss_{\db'}(\param)\right) \right| \\
={} &\max_{\db, \db' \in \univ^{\rows}} \left| \left(\overline{\loss}_{\db}(\overline{\param}) - \min_{\param \in \dom} \overline{\loss}_{\db}(\param)\right) - \left(\overline{\loss}_{\db'}(\overline{\param}) - \min_{\param \in \dom} \overline{\loss}_{\db'}(\param)\right) \right| \\
={} &\max_{\db, \db' \in \univ^{\rows}} \left| \left(\overline{\loss}_{\db}(\overline{\param}) -\loss_{\db'}(\overline{\param}) \right)\right|+ \left| \left(\min_{\param \in \dom} \overline{\loss}_{\db}(\param) - \min_{\param \in \dom} \overline{\loss}_{\db'}(\param)\right) \right| \\
\leq{} &\frac{S}{\rows} + \frac{2S}{\rows} = \frac{3S}{\rows}.
\end{align*}
Since this bound holds for every $\loss \in \lossset$, we have
\begin{align*}
\max_{\loss \in \lossset} \max_{\db, \db' \in \univ^{\rows}} \left| \err_{\loss}(\db, \dbt^{t}) - \err_{\loss}(\db', \dbt^{t}) \right| \leq{} \frac{3S}{\rows}.
\end{align*}

Thus, the queries given to $\cSV$ are indeed $(3S/n)$-sensitive and we are justified in assuming that $\cSV$ is an $(\eps/2,\delta/2)$-differentially private algorithm.

Now, we return to analyzing the privacy loss of $\mech'.$  By assumption, for every fixed $\loss^{t}$, the choice of $\privparamt^{t} = \mech'(\db, \losst^{t})$ is $(\eps_0, \delta_0)$-differentially private with respect to the input $\db$.  Moreover, the choice of $\loss^{t}$ depends only on the output of $\cSV$, which we have already argued is $(\eps/2, \delta/2)$-differentially private.  Therefore, we can view all of the calls to $\mech'$ as a single $T$-fold adaptive composition of $(\eps_0, \delta_0)$-differentially private algorithms.  For $\eps_0, \delta_0$ as specified in the online private multiplicative weights algorithm, the result will be $(\eps/2, \delta/2)$-differentially private.  Since these are the only two ways in which the private dataset $\db$ is used, we have proven that the entire algorithm is $(\eps, \delta)$-differentially private.

\section{Applications of Theorem \ifnum\podscamera=1 2\else \ref{thm:putittogetheronline}\fi}

In this section we give some interpretation of Theorem~\ref{thm:putittogetheronline} and show how it can be applied to specific interesting cases that have been considered in the literature on differentially private convex minimization in order to obtain the results stated in the introduction.

\subsection{Interpreting Theorem~\ref{thm:putittogetheronline}}

In Theorem~\ref{thm:putittogetheronline}, we have assumed that there exists an $(\eps_0, \delta_0)$-differentially private algorithm $\mech'$ that is $(\alpha_0,\beta_0)$-accurate for any one $\loss$ from $\lossset$ given $\rows'$ samples.  By a standard argument, if there exists a $(1, \delta_0)$-differentially private algorithm $\mech''$ that is $(\alpha_0,\beta_0)$-accurate for $\loss$ given $\rows''$ samples, then there exists an $(\eps_0,\delta_0)$-differentially private algorithm with the same accuracy given $O(\rows''/\eps_0)$ samples.  Applying this observation, simplifying, and dropping the dependence on $\beta, \eps, \delta$, we can write the requirement in Theorem~\ref{thm:putittogetheronline} as
\ifnum\podscamera=1
\begin{align*}
\rows \gtrsim{}  &\max \left\{ \frac{\rows''}{\eps_0}, \frac{S^2 \cdot \log k}{\alpha^2}  \right\}  \\
\lesssim{} &\frac{S \cdot \sqrt{\log|\univ|} \cdot \log k}{\alpha} \cdot \max \left\{ \rows'', \frac{S}{\alpha}  \right\}
\end{align*}
\else
\begin{align*}
\rows \gtrsim{}  \max \left\{ \frac{\rows''}{\eps_0}, \frac{S^2 \cdot \log k}{\alpha^2}  \right\} 
\lesssim{} \frac{S \cdot \sqrt{\log|\univ|} \cdot \log k}{\alpha} \cdot \max \left\{ \rows'', \frac{S}{\alpha}  \right\}
\end{align*}
\fi
The first term in the max is just the size of dataset required to answer a single convex minimization query in $\lossset$ with $\eps = 1$.  The second term in the max can be either larger or smaller than $n''$.  However, for the most basic setting of a single, Lipschitz loss function over a bounded domain, $n'' \gg S / \alpha$, so the second term will be dominated by the first term.  

Thus, in some cases, Theorem~\ref{thm:putittogetheronline} can be interpreted as saying that the amount of data required to answer $k$ queries from $\lossset$ is only a factor of $\approx (S \cdot \sqrt{\log|\univ|} \cdot \log k) / \alpha$ larger than the amount of data required to both answer a single query in $\lossset$.  Using the simple composition approach where each of the $k$ queries is answered independently would require a factor of $\approx \sqrt{k}$ more data than answering a single query.  Thus our algorithm is a substantial improvement when $\sqrt{k} \gg (S \cdot \sqrt{\log|\univ|} \cdot \log k)/\alpha$.

\subsection{Applications}

We now show how to instantiate Theorem~\ref{thm:putittogetheronline} with various differentially private algorithms for answering convex minimization queries to obtain the results in the Introduction.  
%In this section, all results are stated in the terminology of Definitions~\ref{def:oneacc} and~\ref{def:manyacc} as the size of dataset $n$ needed to achieve $\alpha$-accuracy.  The results can be translated into the language of Table 1, giving bounds on $\alpha$ as a function of $n$ straightforwardly by rearranging terms.

\subsubsection{Lipschitz and Bounded Loss Functions.}
In much of the work on differentially private convex minimization, the queries are normalized so that the parameter $\param$ lies in a unit $L_2$ ball, and the loss function $\loss$ satisfies a Lipschitz condition.  Bassily, Smith, and Thakurta~\cite{BassilyST14} recently showed optimal upper and lower bounds for answering a single query from this family.  Formally,
\begin{theorem}[\cite{BassilyST14}] \label{thm:oneboundedlipschitz}
Let $\ell \from \dom \times \univ \to \R$ be a convex loss function where
$
\dom \subseteq \set{\param \in \R^d \mid \, \|\param\|_2 \leq 1}
$
and for every $\param \in \dom$, $\row \in \univ$,
$
\| \grad \loss_\row(\param) \|_2 \leq 1.
$  Let $\query_{\loss}$ be the associated CM query.
There is a $(\eps_0,\delta_0)$-differentially private algorithm that is $(\alpha_0, \beta_0)$-accurate for $\query_{\ell}$ on datasets of size $n$ for
$$
\rows = O\left(\frac{\sqrt{d}}{\alpha_0 \eps_0}\right) \cdot \polylog\left(\frac{1}{\delta_{0}}, \frac{1}{\beta_{0}}\right).
$$  
\end{theorem}

Note that if $\dom$ is contained in a unit $L_2$ ball and $\loss$ is $1$-Lipschitz, then the scaling parameter $S$ is at most $2$.  Combining Theorem~\ref{thm:putittogetheronline} and Theorem~\ref{thm:oneboundedlipschitz} yields the following result.

\begin{theorem}
Let $\lossset$ be the set of convex loss functions $\ell \from \dom \times \univ \to \R$ for 
$
\dom \subseteq \set{\param \in \R^d \mid \, \|\param\|_2 \leq 1}
$
such that for every $\loss \in \lossset$, $\param \in \dom$, $\row \in \univ$,
$
\| \grad \loss_\row(\param) \|_2 \leq 1.
$
Let $\queryset_{\lossset}$ be the associated family of CM queries.
There is an $(\eps, \delta)$-differentially private algorithm that is $(\alpha, \beta)$-accurate for $k$ CM queries from $\queryset_{\lossset}$ on datasets of size $n$ for
$$
\rows = \tilde{O}\left(\frac{\sqrt{\log|\univ|}}{\alpha^2 \eps} \cdot \max\left\{\sqrt{d}, \log k \right\}\right) \cdot \polylog\left(\frac{1}{\delta}, \frac{1}{\beta}\right).
$$
\end{theorem}

\subsubsection{Generalized Linear Models.}
Using the algorithm of Theorem~\ref{thm:oneboundedlipschitz}, $n$ must grow polynomially with $d$ to solve even a single CM query in dimension $d$, and this was shown to be inherent by Bassily et al.~\cite{BassilyST14} (building on~\cite{BunUV14}).  However, the work of Jain and Thakurta~\cite{JainT14} shows that dependence on $d$ can be avoided for the important class of \emph{unconstrained generalized linear models}.  For example, logistic regression and linear regression are generalized linear models.  A convex loss function $\loss \from \dom \times \univ \to \R$ is a generalized linear model if $\dom \subseteq \R^d$, $\univ \subseteq \R^d$, and
$\loss(\param, \row)$ depends only on the inner product of $\param$ and $\row$.  That is, there exists a convex function $\loss' \from \R \to \R$ such that
$
\loss(\param,\row) = \loss'(\langle \param, \row \rangle).
$
We say that the generalized linear model is unconstrained if there are no constraints other than boundedness.  That is, $\dom = \set{\param \in \R^d \mid \|\param\|_2 \leq 1}$.
\begin{theorem}[\cite{JainT14}] \label{thm:oneglm}
Let $\ell \from \dom \times \univ \to \R$ be an unconstrained generalized linear model with the domain
$
\dom = \set{\param \in \R^d \mid \, \|\param\|_2 \leq 1}
$ and for every $\param \in \dom$, $\row \in \univ$,
$
\| \grad \loss_\row(\param) \|_2 \leq 1.
$
Let $\query_{\loss}$ be the associated CM query.
There is a $(\eps_0,\delta_0)$-differentially private algorithm that is $(\alpha_0, \beta_0)$-accurate for $\query_{\ell}$ on datasets of size $n$ for
$$
\rows = O\left(\frac{1}{\alpha_0^2 \eps_0}\right) \cdot \polylog\left(\frac{1}{\delta_{0}}, \frac{1}{\beta_{0}}\right).
$$  
\end{theorem}
Combining Theorem~\ref{thm:putittogetheronline} and Theorem~\ref{thm:oneglm} yields the following result.
\begin{theorem}
Let $\lossset$ be the set of unconstrained generalized linear models $\ell \from \dom \times \univ \to \R$ with the domain
$
\dom \subseteq \set{\param \in \R^d \mid \, \|\param\|_2 \leq 1}
$
such that for every $\loss \in \lossset$, $\param \in \dom$, $\row \in \univ$,
$
\| \grad \loss_\row(\param) \|_2 \leq 1.
$
Let $\queryset_{\lossset}$ be the associated family of CM queries.
There is an $(\eps, \delta)$-differentially private algorithm that is $(\alpha, \beta)$-accurate for $k$ CM queries from $\queryset_{\lossset}$ given $\rows$ records for
$$
\rows = \tilde{O}\left(\frac{\sqrt{\log|\univ|}}{\alpha^2 \eps} \cdot \max\left\{\frac{1}{\alpha}, \log k \right\}\right) \cdot \polylog\left(\frac{1}{\delta}, \frac{1}{\beta}\right).
$$
\end{theorem}

\subsubsection{Strongly Convex Loss Functions.}  
Stronger accuracy guarantees for answering a single CM query are also achievable in the common special case where $\loss$ is strongly convex.  Informally, $\loss$ is strongly convex if it can be lower bounded by a quadratic function.  Specifically, for a parameter $\sigma \geq 0$, the function $\loss \from \dom \times \univ \to \R$ is $2\sigma$-strongly convex if for every $\param, \param' \in \dom$ and $\row \in \univ$, $\loss(\param'; \row) \geq \loss(\param; \row) + \langle \param' - \param, \grad \loss(\param; \row) \rangle +  \sigma \| \param' - \param \|_2^2$.  In the previous statement, the gradient is with respect to $\param$.
\begin{theorem}[\cite{BassilyST14}] \label{thm:onestronglyconvex}
Let $\ell \from \dom \times \univ \to \R$ be a $\sigma$-strongly convex loss function where
$
\dom \subseteq \set{\param \in \R^d \mid \, \|\param\|_2 \leq 1}
$
and for every $\param \in \dom$, $\row \in \univ$,
$
\| \grad \loss_\row(\param) \|_2 \leq 1.
$  Let $\query_{\loss}$ be the associated CM query.
There is a $(\eps_0,\delta_0)$-differentially private algorithm that is $(\alpha_0, \beta_0)$-accurate for $\query_{\ell}$ on datasets of size $n$ for
$$
\rows = O\left(\frac{\sqrt{d}}{\sqrt{\sigma \alpha_0} \eps_0}\right) \cdot \polylog\left(\frac{1}{\delta_{0}}, \frac{1}{\beta_{0}}\right).
$$  
\end{theorem}
Combining Theorem~\ref{thm:putittogetheronline} and Theorem~\ref{thm:onestronglyconvex} yields the following result.
\begin{theorem}
Let $\lossset$ be the set of $\sigma$-strongly convex loss functions $\ell \from \dom \times \univ \to \R$ for
$
\dom \subseteq \set{\param \in \R^d \mid \, \|\param\|_2 \leq 1}
$
such that for every $\loss \in \lossset$, $\param \in \dom$, $\row \in \univ$,
$
\| \grad \loss_\row(\param) \|_2 \leq 1.
$
Let $\queryset_{\lossset}$ be the associated family of CM queries.
There is an $(\eps, \delta)$-differentially private algorithm that is $(\alpha, \beta)$-accurate for $k$ CM queries from $\queryset_{\lossset}$ on datasets of size $n$ for
$$
\rows = \tilde{O}\left(\frac{\sqrt{\log |\univ|}}{\eps} \max\left\{ \frac{\sqrt{d}}{\sqrt{\sigma} \alpha^{3/2}} , \frac{\log k}{\alpha^2} \right\}\right) \cdot \polylog\left(\frac{1}{\delta}, \frac{1}{\beta}\right)
$$
\end{theorem}

\subsection{Running Time and Discussion of \ifnum\podscamera=1 \\ \else \fi Computational Complexity} \label{sec:complexity}
In this section we discuss the computational complexity of the algorithm.  To do so, we assume $\dom \subseteq \R^d$, and for simplicity and concreteness we consider the natural choice of data universe $\univ = \bits^{d}$, or equivalently $\univ = \set{\frac{\pm1}{\sqrt{d}}}^d$.  Since our algorithm uses the ability to solve a single CM query in $\lossset$ as a blackbox, we assume that this step can be done in $\poly(\rows, d)$ time both privately and non-privately.  For this informal discussion, we also ignore the dependence in running time on $S, \alpha, \beta, \eps, \delta$, which will not substantially affect the conclusions.  

There are three main steps that dominate the running time of each of the $k$ iterations:
\begin{enumerate}
\item Running the online sparse vector algorithm $\cSV$ on $q^j$.  This step can be done in time $\poly(n, d)$.
\item If $a^j = \top$, finding a private approximate minimizer of $\losst^{j}$.  By assumption, this step can be done in time $\poly(\rows, d)$.
\item If $a^j = \top$, computing the new histogram $\dbt^{t+1}$.  This step requires time $\tilde{O}(2^d)$.
\end{enumerate}

Since each of these steps is carried out for $k$ steps, the overall running time is $\poly(\rows, 2^d, k)$.  Even tough it was useful to think of the database as a histogram, which is a vector of length $2^d$, the input database $\db$ would more naturally be represented as a collection of records $\db \in (\bits^{d})^\rows$.  Thus it is natural to look for an algorithm with running time $\poly(\rows, d, k)$.  In summary, even when the individual loss functions can be privately minimized in $\poly(\rows, d)$ time, our algorithm requires time $\poly(\rows, 2^d, k)$, which is exponential in the dimension of the data.  More generally, there is a polynomial dependence on $|\univ|$, where one would hope for a polylogarithmic dependence.

Unfortunately, this exponential running time is inherent.  Since CM queries generalize the well studied class of linear queries, we can carry over the hardness results of Ullman~\cite{Ullman13} to this setting.  Specifically, assuming the existence of one-way functions, there is no $\poly(\rows, d)$-time algorithm that takes as input a set of $k$ arbitrary differentiable convex loss functions, and a database $\db \in (\bits^{d})^n$ for $n \leq k^{1/2 - o(1)}$, and and outputs answers that are even $1/100$-accurate for each query in $\lossset$.

Although the hardness result rules out an efficient mechanism for answering an arbitrary large set of CM queries, more efficient algorithms may be possible for specific families $\lossset$.  In the setting of counting queries, such algorithms are known for special cases such as \emph{interval queries}~\cite{BeimelNS13} and \emph{marginal queries}~\cite{GuptaHRU11, HardtRS12, ThalerUV12, ChandrasekaranTUW14, DworkNT13}.  It would be interesting to see if techniques from those works can be applied to give more efficient algorithms for natural families of CM queries.  We remark that Ullman and Vadhan~\cite{UllmanV11} show that efficient algorithms that output synthetic data cannot be accurate even for very simple families of counting queries, and thus also for certain very simple families of CM queries.  Our algorithm indeed can be modified to output a synthetic dataset (namely, the final histogram $\dbt^{t}$ used in the execution of the algorithm), and thus substantially different techniques would be required to answer interesting classes of CM queries more efficiently.  We leave it as an interesting direction for future work to improve the running time of our algorithm for interesting restricted families of CM queries.

\addcontentsline{toc}{section}{Acknowledgements}
\section*{Acknowledgements}
\ifnum\podscamera=0
We thank Adam Smith and Salil Vadhan for helpful discussions.
\else
We would like to thank Adam Smith and Salil Vadhan for helpful discussions.
\fi

\ifnum\podscamera=1
	\pagebreak
	\bibliographystyle{abbrv}
	\bibliography{references} 
\else
	\addcontentsline{toc}{section}{References}
	\bibliographystyle{alpha}
	\bibliography{references}

\newcommand{\etalchar}[1]{$^{#1}$}
\begin{thebibliography}{CTUW14}

\bibitem[AHK12]{AroraHK12}
Sanjeev Arora, Elad Hazan, and Satyen Kale.
\newblock The multiplicative weights update method: a meta-algorithm and
  applications.
\newblock {\em Theory of Computing}, 8(1):121--164, 2012.

\bibitem[BLR08]{BlumLR08}
Avrim Blum, Katrina Ligett, and Aaron Roth.
\newblock A learning theory approach to non-interactive database privacy.
\newblock In {\em ACM Symposium on Theory of Computing (STOC '08)}, pages
  609--618. ACM, 17-20 May 2008.

\bibitem[BNS13]{BeimelNS13}
Amos Beimel, Kobbi Nissim, and Uri Stemmer.
\newblock Private learning and sanitization: Pure vs. approximate differential
  privacy.
\newblock In {\em APPROX-RANDOM}, pages 363--378. Springer, 21-23 August 2013.

\bibitem[BSSU15]{BassilySSU15}
Raef Bassily, Adam Smith, Thomas Steinke, and Jonathan Ullman.
\newblock More general queries with better generalization error in adaptive
  data analysis.
\newblock Manuscript, 2015.

\bibitem[BST14]{BassilyST14}
Raef Bassily, Adam Smith, and Abhradeep Thakurta.
\newblock Private empirical risk minimization, revisited.
\newblock {\em CoRR}, abs/1405.7085, 2014.

\bibitem[BUV14]{BunUV14}
Mark Bun, Jonathan Ullman, and Salil~P. Vadhan.
\newblock Fingerprinting codes and the price of approximate differential
  privacy.
\newblock In {\em ACM Symposium on Theory of Computing (STOC '14)}. ACM, 1--3
  June 2014.

\bibitem[CMS11]{ChaudhuriMS11}
Kamalika Chaudhuri, Claire Monteleoni, and Anand~D. Sarwate.
\newblock Differentially private empirical risk minimization.
\newblock {\em Journal of Machine Learning Research}, 12:1069--1109, 2011.

\bibitem[CTUW14]{ChandrasekaranTUW14}
Karthekeyan Chandrasekaran, Justin Thaler, Jonathan Ullman, and Andrew Wan.
\newblock Faster private release of marginals on small databases.
\newblock In {\em Innovations in Theoretical Computer Science (ITCS '14)},
  pages 387--402. ACM, 12-14 January 2014.

\bibitem[DFH{\etalchar{+}}15]{DworkFHPRR15}
Cynthia Dwork, Vitaly Feldman, Moritz Hardt, Toniann Pitassi, Omer Reingold,
  and Aaron Roth.
\newblock Preserving statistical validity in adaptive data analysis.
\newblock In {\em STOC}. {ACM}, June 14--17 2015.

\bibitem[DL09]{DworkL09}
Cynthia Dwork and Jing Lei.
\newblock Differential privacy and robust statistics.
\newblock In {\em ACM Symposium on Theory of Computing (STOC '09)}, pages
  371--380. ACM, 31 May - 2 June 2009.

\bibitem[DMNS06]{DworkMNS06}
Cynthia Dwork, Frank McSherry, Kobbi Nissim, and Adam Smith.
\newblock Calibrating noise to sensitivity in private data analysis.
\newblock In {\em Theory of Cryptography (TCC '06)}, pages 265--284. Springer,
  4--7 March 2006.

\bibitem[DNR{\etalchar{+}}09]{DworkNRRV09}
Cynthia Dwork, Moni Naor, Omer Reingold, Guy~N. Rothblum, and Salil~P. Vadhan.
\newblock On the complexity of differentially private data release: efficient
  algorithms and hardness results.
\newblock In {\em ACM Symposium on Theory of Computing (STOC '09)}, pages
  381--390. ACM, 31 May - 2 June 2009.

\bibitem[DNT13]{DworkNT13}
Cynthia Dwork, Aleksandar Nikolov, and Kunal Talwar.
\newblock Efficient algorithms for privately releasing marginals via convex
  relaxations.
\newblock {\em CoRR}, abs/1308.1385, 2013.

\bibitem[DR14]{DworkR14}
Cynthia Dwork and Aaron Roth.
\newblock The algorithmic foundations of differential privacy.
\newblock {\em Foundations and Trends in Theoretical Computer Science},
  9(3-4):211--407, 2014.

\bibitem[DRV10]{DworkRV10}
Cynthia Dwork, Guy~N. Rothblum, and Salil~P. Vadhan.
\newblock Boosting and differential privacy.
\newblock In {\em IEEE Symposium on Foundations of Computer Science (FOCS
  '10)}, pages 51--60. IEEE Computer Society, 23-26 October 2010.

\bibitem[GHRU11]{GuptaHRU11}
Anupam Gupta, Moritz Hardt, Aaron Roth, and Jonathan Ullman.
\newblock Privately releasing conjunctions and the statistical query barrier.
\newblock In {\em ACM Symposium on Theory of Computing (STOC '11)}, pages
  803--812. ACM, 6-8 June 2011.

\bibitem[GRU12]{GuptaRU12}
Anupam Gupta, Aaron Roth, and Jonathan Ullman.
\newblock Iterative constructions and private data release.
\newblock In {\em Theory of Cryptography (TCC '12)}, pages 339--356. Springer,
  19-21 March 2012.

\bibitem[HLM12]{HardtLM12}
Moritz Hardt, Katrina Ligett, and Frank McSherry.
\newblock A simple and practical algorithm for differentially private data
  release.
\newblock In {\em Neural Information Processing Systems (NIPS '12)}, pages
  2348--2356, 3-6 December 2012.

\bibitem[HR10]{HardtR10}
Moritz Hardt and Guy~N. Rothblum.
\newblock A multiplicative weights mechanism for privacy-preserving data
  analysis.
\newblock In {\em IEEE Symposium on Foundations of Computer Science (FOCS
  '10)}, pages 61--70. IEEE Computer Society, 23-26 October 2010.

\bibitem[HRS12]{HardtRS12}
Moritz Hardt, Guy~N. Rothblum, and Rocco~A. Servedio.
\newblock Private data release via learning thresholds.
\newblock In {\em ACM-SIAM Symposium on Discrete Algorithms (SODA '12)}, pages
  168--187. SIAM, 17-19 January 2012.

\bibitem[HU14]{HardtU14}
Moritz Hardt and Jonathan Ullman.
\newblock Preventing false discovery in interactive data analysis is hard.
\newblock In {\em FOCS}. IEEE, October 19-21 2014.

\bibitem[JT14]{JainT14}
Prateek Jain and Abhradeep~Guha Thakurta.
\newblock (near) dimension independent risk bounds for differentially private
  learning.
\newblock In {\em ICML}, pages 476--484. JMLR.org, 21-26 June 2014.

\bibitem[KRS13]{KasiviswanathanRS13}
Shiva~Prasad Kasiviswanathan, Mark Rudelson, and Adam Smith.
\newblock The power of linear reconstruction attacks.
\newblock In {\em SODA}, pages 1415--1433. SIAM, 6-8 Jan 2013.

\bibitem[KST12]{KiferST12}
Daniel Kifer, Adam~D. Smith, and Abhradeep Thakurta.
\newblock Private convex optimization for empirical risk minimization with
  applications to high-dimensional regression.
\newblock In {\em Conference on Learning Theory (COLT '12)}, pages 25.1--25.40.
  JMLR.org, 25-27 June 2012.

\bibitem[MT07]{McSherryT07}
Frank McSherry and Kunal Talwar.
\newblock Mechanism design via differential privacy.
\newblock In {\em FOCS}, pages 94--103. IEEE Computer Society, 20-23 October
  2007.

\bibitem[RR10]{RothR10}
Aaron Roth and Tim Roughgarden.
\newblock Interactive privacy via the median mechanism.
\newblock In {\em ACM Symposium on Theory of Computing (STOC '10)}, pages
  765--774. ACM, 5-8 June 2010.

\bibitem[TS13]{SmithT13}
Abhradeep Thakurta and Adam Smith.
\newblock Differentially private feature selection via stability arguments, and
  the robustness of the lasso.
\newblock In {\em Conference on Learning Theory (COLT '13)}, pages 819--850.
  JMLR.org, 12-14 June 2013.

\bibitem[TUV12]{ThalerUV12}
Justin Thaler, Jonathan Ullman, and Salil~P. Vadhan.
\newblock Faster algorithms for privately releasing marginals.
\newblock In {\em International Colloquium on Automata, Languages, and
  Programming (ICALP '12)}, pages 810--821. Springer, 9-13 July 2012.

\bibitem[Ull13]{Ullman13}
Jonathan Ullman.
\newblock Answering n$^{2+o(1)}$ counting queries with differential privacy is
  hard.
\newblock In {\em ACM Symposium on Theory of Computing (STOC '13)}, pages
  361--370. ACM, 1-4 June 2013.

\bibitem[UV11]{UllmanV11}
Jonathan Ullman and Salil~P. Vadhan.
\newblock {PCP}s and the hardness of generating private synthetic data.
\newblock In {\em Theory of Cryptography (TCC '11)}, pages 400--416. Springer,
  28-30 March 2011.

\end{thebibliography}
\fi

\end{document}